\documentclass[twoside,leqno]{article}

\usepackage{ltexpprt}

\usepackage[utf8]{inputenc} 
\usepackage[T1]{fontenc}    
\usepackage[hyperindex,breaklinks,bookmarks,hidelinks]{hyperref}       
\usepackage{url}            
\usepackage{booktabs}       
\usepackage{amsfonts}       
\usepackage{nicefrac}       
\usepackage{microtype}      
\usepackage{xcolor}         

\usepackage[ruled,linesnumbered,noend,algo2e]{algorithm2e}
\usepackage{float}
\usepackage{mathtools}

\newcommand{\normal}[1]{\mathcal{N}(0,#1)}

\usepackage{todonotes}

\usepackage[english]{babel}

\newtheorem{@definition}{Definition}[section]
 \newenvironment{definition}{\begin{@definition}}{\end{@definition}}

\date{}
\title{Better Gaussian Mechanism using Correlated Noise}

\author{
    Christian Janos Lebeda\thanks{This work was carried out when Lebeda was at the IT University of Copenhagen.} \\
    Inria, University of Montpellier\\
    \texttt{christian-janos.lebeda@inria.fr}
}

\begin{document}

\maketitle

\begin{abstract}
  We present a simple variant of the Gaussian mechanism for answering differentially private queries when the sensitivity space has a certain common structure.
  Our motivating problem is the fundamental task of answering $d$ counting queries under the add/remove neighboring relation.
  The standard Gaussian mechanism solves this task by adding noise distributed as a Gaussian with variance scaled by $d$ independently to each count.
  We show that adding a random variable distributed as a Gaussian with variance scaled by $(\sqrt{d} + 1)/4$ to all counts allows us to reduce the variance of the independent Gaussian noise samples to scale only with $(d + \sqrt{d})/4$.
  The total noise added to each counting query follows a Gaussian distribution with standard deviation scaled by $(\sqrt{d} + 1)/2$ rather than $\sqrt{d}$.
  The central idea of our mechanism is simple and the technique is flexible.
  We show that applying our technique to another problem gives similar improvements over the standard Gaussian mechanism.
\end{abstract}

\section{Introduction}
\label{sec:intro}

The Gaussian mechanism is one of the most important tools in differential privacy.
In its most basic form, the mechanism first computes the answer to $d$ real-valued queries.
The output is then perturbed by adding noise sampled from a zero-mean Gaussian distribution independently to the answer of each query.
The magnitude of the Gaussian noise required to satisfy differential privacy depends on the desired privacy parameters and the $\ell_2$ sensitivity of the queries. 
The mechanism is used as a key subroutine in a long list of more complicated differential privacy mechanisms such as DP-SGD~\cite{deepLearningDP} and the Census TopDown Algorithm~\cite{censusTopDown} to name just two examples (The TopDown Algorithm uses the discrete variant of the Gaussian mechanism~\cite{discrete-Gaussian}). 

The Gaussian mechanism was introduced in some of the earliest work on differential privacy~\cite{dworkCalibratingNoise,DworkKMMN06OurDataOurselves}.
Since then, the mechanism has been subject to a lot of study due to its importance. 
Notable results include a tight analysis under approximate differential privacy~\cite{analyticalGaussian} as well as other definitions of differential privacy~\cite{zeroConcentrated, renyiDP,gaussianDPand-f-DP}.
An adjacent line of work has improved the analysis of mechanisms where additive Gaussian noise is a key element. 
Recent results include an exact analysis of the Gaussian Sparse Histogram Mechanism~\cite{GaussianSparseHistogramMechanism} and development of numerical tools that allow for tighter privacy accounting used for DP-SGD~(e.g. \cite{sommerPLD,koskela20-gaussian-variants,koskela21-one-dimensional-outputs,zhu22-optimal-characteristic-functions}).
These results give stronger privacy guarantees of mechanisms for various settings, which imply an improved privacy-utility trade-off.
Improving the Gaussian mechanism and its variants by even small constant factors can have a significant impact on downstream applications.
One notable such example is the accuracy of models trained with DP-SGD which can vary significantly based on the amount of noise
(see e.g.~\cite[Figure~1]{DeBHSB22}).

As mentioned above, the standard Gaussian mechanism privately releases queries by adding noise with magnitude scaled to the $\ell_2$ sensitivity.
The $\ell_2$ sensitivity is defined as the biggest distance in the $\ell_2$ norm between the query outputs for any pair of neighboring datasets.
If we have no guarantees on the $\ell_2$ norm of the query outputs we can clip data points before adding noise.
This is common practice for several tasks such as approximating gradients for private machine learning~(e.g. \cite{deepLearningDP,DP-FTRL}) and mean estimation~(e.g. \cite{biswas_coinpress_2020, huang_instance-optimal_2021, PLAN}). 
The Gaussian mechanism can be applied to many settings because we only need to bound the $\ell_2$ sensitivity to privately release an estimate to any real-value queries.
However, in some cases we can achieve better accuracy 
by utilizing the structure of the queries.

In this work we focus on settings where we want to estimate the sum of all data points in a dataset.
We refer to the set of all possible differences between the non-private sums for neighboring dataset as the sensitivity space.
We utilize a simple structure of the sensitivity space present for certain queries.
Our technique reduces the magnitude of noise for such queries over the standard Gaussian mechanism that only considers the maximum $\ell_2$ norm of the sensitivity space.
The structure we consider most naturally occurs under the add/remove neighboring relation.

Before presenting our technique it is worth considering a query where the structure we utilize is not present.
We compare the sensitivity under the add/remove and replacement neighboring relations.
For many tasks, the $\ell_2$ sensitivity under the replacement neighboring relation is twice that of the add/remove relation.
For example, if we clip the $\ell_2$ norm of data points in $\mathbb{R}^d$ to some value $C > 0$ the $\ell_2$ sensitivity is clearly $C$ under the add/remove relation.
But the $\ell_2$ sensitivity is $2C$ under replacement because we can replace any vector $x \in \mathbb{R}^d$ where $\| x \|_2 = C$ with $-x$.
The magnitude of noise added by the standard Gaussian mechanism therefore differs by a factor of 2 between the two neighboring relations.

In contrast to the example above, consider the task of answering $d$ counting queries. 
Here each data point is in $\{0,1\}^d$.
The $\ell_2$ distance between neighboring datasets under add/remove is maximized when adding or removing the vector of all ones (which we denote ${\bf 1}^d$).
Under the replacement neighboring relation the distance is maximized by replacing any data point with one that disagrees on all coordinates.
In both cases the $\ell_2$ sensitivity is $\sqrt{d}$.
The standard Gaussian mechanism therefore adds noise with the same magnitude under either neighboring relation.
Our main result is to show that adding the same sample from a Gaussian distribution to all counts allows us to reduce the magnitude of the independent noise added to each query under the add/remove neighboring relation.
\medskip
\\ \noindent 
\textsc{Theorem~\ref{thm:main-theorem} (simplified).} 
\textit{
Let $f(X) \in \mathbb{N}^d$ be the answer to $d$ counting queries on dataset $X$.
Define random variables such that $\eta \sim \normal{\frac{\sqrt{d} + 1}{4\mu^2}}$, $Z \sim \normal{\frac{d + \sqrt{d}}{4\mu^2}I_d}$, and $W \sim \normal{\frac{d}{\mu^2}I_d}$ for any $\mu > 0$.
Then the mechanism that outputs $f(X) + \eta{\bf 1}^d + Z$ has the same differential privacy guarantees under the add/remove neighboring relation as the mechanism that outputs $f(X) + W$.} 
\medskip

The value of $\mu$ depends on the desired privacy guarantees.
For simplicity of presentation we assume here that $\mu = 1$.
The error of each counting query in our mechanism depends on the sum of two zero-mean Gaussian noise samples. 
The error is distributed as $\normal{\frac{d + 2\sqrt{d} + 1}{4}}$ and as such
the standard deviation of our mechanism is $(\sqrt{d} + 1)/2$ while the standard deviation for the standard Gaussian mechanism is $\sqrt{d}$.
Note that in general we cannot hope to add Gaussian noise under add/remove for any query with less than half the magnitude required under replacement.
The group privacy property of differential privacy allows us to use any mechanism designed for add/remove under the replacement relation (see e.g.~\cite{gaussianDPand-f-DP,zeroConcentrated}).
Our mechanism adds near-optimal Gaussian noise even for a more restrictive setting which we discuss in Section~\ref{sec:bounded-density}.

Many data analysis tasks depend on the size of the input dataset $n$.
However, since the size of neighboring datasets differ by $1$ under the add/remove neighboring relation we cannot release the exact size under differential privacy.
This is often disregarded as a technical detail because we can estimate the size of the dataset using additional privacy budget.
Some private algorithms explicitly assume that we have access to either an estimate or a bound of $n$ (see e.g. \cite{propose-test-release,geometryOfDPSparseApproximate}).
Other mechanisms implicitly rely on an estimate of the dataset size such as some versions of DP-SGD~\cite{deepLearningDP} since gradients are usually normalized based on the expected batch size which is a function of $n$.
Our mechanism also relies on an estimate of the dataset size but we do not assume that an estimate is known.
We instead use part of our privacy budget to compute an estimate.
There is a trade-off between the accuracy of this estimate and the error for each query which we discuss in Section~\ref{sec:reduced-correlated-noise}.
If we have access to an estimate of the dataset size we can reduce the magnitude of noise added to each query from $(\sqrt{d} + 1)/2$ to $\sqrt{d}/2$.

Our mechanism is simple and it is likely to be easily implementable in libraries that already support the standard Gaussian mechanism. 
We argue that simplicity is a desirable feature of practical differentially private mechanisms.
In many cases (e.g. the 2020 US Census~\cite{censusTopDown}) the privately released data will be used by data analysts who might not be familiar with the theory behind differential privacy.
It is easier to communicate how the data was protected when the differentially private mechanism is simple. 
It is also easier to take the effect of introduced noise into account for any downstream analysis.
Furthermore, it is well known that implementations of even basic concepts and mechanisms from differential privacy requires care to avoid privacy issues~(See e.g.~\cite{mironov-floating-point-issues,CasacubertaSVW22sensitivity,DP-floating-point-attacks}).
We discuss some considerations for implementation at the end of Section~\ref{sec:algo}.
We believe that our technique has the potential to improve mechanisms for other settings apart from counting queries.
In Section~\ref{sec:extensions} 
we adapt our core idea to a setting where we add noise with a smaller magnitude than the standard Gaussian mechanism under both the add/remove and replacement neighboring relations.
We conclude our paper by discussing related work in Section~\ref{sec:related} and potential future directions in Section~\ref{sec:conclusion}.

\section{Preliminaries}
\label{sec:prelim}

{\it Setup of this paper.}
The Gaussian mechanism is used to privately estimate any real-valued function of a dataset $X \in \mathcal{U}^\mathbb{N}$.
We focus on settings where a dataset $X = (x_1, x_2, \dots, x_n)$ consists of $n$ data points where $x_i \in \mathbb{R}^d$ and we want to estimate the sum of all data points $f(X) = \sum_{i = 1}^n x_i$. 
Throughout this paper we write $[m]$ to denote the set $\{1,2,\dots,m\}$ for any $m \in \mathbb{N}$, we denote by ${\bf r}^d$ the $d$ dimensional vector where all entries equal some value $r \in \mathbb{R}$, and we write $I_d$ to denote the $d \times d$ identity matrix. 

{\it Differential Privacy.}
Differential privacy~\cite{dworkCalibratingNoise} is a statistical property of a randomized mechanism.
Intuitively, a randomized mechanism satisfies differential privacy if the output distribution 
of the mechanism is not affected significantly by any individual.
The notion of neighboring datasets is used to model the impact of any one individual's data.
The two most commonly used definitions are the add/remove and replacement variants defined below.
We consider both definitions in this paper but focus mainly on the add/remove neighboring relation.
\begin{definition}[Neighboring datasets (denoted $X \sim X'$)]
    \label{def:neighboring-datasets}
    Let $X = (x_1, x_2, \dots, x_n)$ be a dataset of size $n$. 
    If $X$ and $X'$ are neighboring datasets under the add/remove neighboring relation then there exists an $i$ such that either $X' = (x_1, \dots, x_{i - 1}, x_{i + 1}, \dots, x_n)$ or $X = (x'_1, \dots, x'_{i - 1}, x'_{i + 1}, \dots, x_{n + 1})$. 
    If $X$ and $X'$ are neighboring datasets under the replacement neighboring relation we have $|X| = |X'|$ and there exists an $i$ such that $x_j = x'_j$ for all $j \neq i$.
\end{definition}

We present our mechanism in terms of Gaussian differential privacy ($\mu$-GDP)~\cite{gaussianDPand-f-DP}.
Informally, an algorithm satisfies $\mu$-GDP if the output distributions for any neighboring datasets are at least as similar as two Gaussian distributions.
This is more precisely defined in Definition~\ref{def:gdp}.
Note that none of our results relies on Gaussian differential privacy. 
In the remark below we discuss how to easily translate our results to some other popular definitions of differential privacy by replacing $\mu$ with an appropriate value.

\begin{definition}[{Gaussian differential privacy~\cite[Definition~4]{gaussianDPand-f-DP}}]
    \label{def:gdp}
    For any randomized mechanism $\mathcal{M} \colon \mathcal{U}^\mathbb{N} \rightarrow \mathcal{R}$ satisfying $\mu$-GDP it holds for all neighboring datasets $X \sim X'$ that
    \[
        T(\mathcal{M}(X),\mathcal{M}(X')) \geq T(\normal{1}, \mathcal{N}(\mu, 1)) \enspace .
    \]
\end{definition}

Here $T(P,Q): [0,1] \rightarrow [0,1]$ denotes a trade-off function as defined below, and the inequality in Definition~\ref{def:gdp} must hold for all inputs to the functions (any $0 \leq \alpha \leq 1$).
The trade-off function gives a lower bound on the type II error achievable for any level of type I error.
We note that understanding the technical details of the trade-off function and $\mu$-GDP is not required to follow the proofs in this paper.
For more background on trade-off functions and $f$-differential privacy, we refer interested readers to Definitions~2 and~3 of~\cite{gaussianDPand-f-DP}.

\begin{definition}[Trade-off function]
    Let $P$ and $Q$ be any two probability distributions defined on the same space. The trade-off function $T(P,Q): [0,1] \rightarrow [0,1]$ is defined as
    \[
        T(P,Q)(\alpha) = \inf\{\beta_\phi : \alpha_\phi \leq \alpha \} \,,
    \]
    where the infimum is taken over all (measurable) rejection rules. For a rejection rule $\phi$, the values of $\alpha_\phi$ and $\beta_\phi$ denotes the type I and type II error rates, respectively. 
\end{definition}

\paragraph*{Remark}
    Let $f \colon \mathcal{U}^\mathbb{N} \rightarrow \mathbb{R}$ be a function with sensitivity $1$.
    Let $\sigma > 0$ be large enough such that $f(X) + \mathcal{N}(0, \sigma^2)$ satisfies either $(\varepsilon, \delta)$-differential privacy~\cite[Theorem~8]{analyticalGaussian}, $(\alpha, \varepsilon)$-RDP~\cite[Corollary 3]{renyiDP}, or $\rho$-zCDP~\cite[Proposition~1.6]{zeroConcentrated}.
    Then all mechanisms presented in this paper satisfy the desired differential privacy guarantee if $1/\mu^2$ is replaced by $\sigma^2$.
\\

The $\ell_2$ sensitivity of a function as defined below is a bound on the distance between the output for any pair of neighboring datasets. 
More generally, the sensitivity space is the set of all possible differences between outputs for neighboring datasets.
The $\ell_2$ sensitivity equals the maximum $\ell_2$ norm for any element in the sensitivity space.

\begin{definition}[$\ell_2$ sensitivity]
  \label{def:l2-sens}
  The $\ell_2$ sensitivity of a function $f \colon \mathcal{U}^\mathbb{N} \rightarrow \mathbb{R}^d$ is
  \[
    \Delta f = \max_{X \sim X'} \|f(X) - f(X')\|_2 \,,
  \]
  where $\|x\|_2 = \sqrt{\sum_{i=1}^d x_i^2}$ denotes the $\ell_2$ norm of any $x \in \mathbb{R}^d$.
\end{definition}

\begin{definition}[Sensitivity space]
    The sensitivity space of a function $f \colon \mathcal{U}^\mathbb{N} \rightarrow \mathbb{R}^d$ is the set
    \[
        \{f(X) - f(X') \in \mathbb{R}^d \vert X \sim X' \} \,.
    \]
\end{definition}

The magnitude of noise added by the Gaussian mechanism scales linearly in the $\ell_2$ sensitivity and a value based on the desired privacy guarantees.
Lemma~\ref{lem:gaussian-mech-GDP} gives the privacy guarantees of the Gaussian mechanism.
We refer to this variant as the standard Gaussian mechanism throughout the paper.

\begin{lemma}[The Gaussian mechanism~\cite{gaussianDPand-f-DP}]
  \label{lem:gaussian-mech-GDP}
  Let $f \colon \mathcal{U}^\mathbb{N} \rightarrow \mathbb{R}^d$ be a function with $\ell_2$ sensitivity $\Delta f$. 
  Then the mechanism that outputs $f(X) + Z$ where $Z \sim \mathcal{N}(0, (\Delta f)^2/\mu^2 I_d)$ satisfies $\mu$-GDP.
\end{lemma}

The last tool we need from the differential privacy literature is the fact that the privacy guarantees of a mechanism are never negatively affected by post-processing.

\begin{lemma}[Post-processing~{\cite[Proposition~4]{gaussianDPand-f-DP}}]
    \label{lem:post-processing}
    Let $\mathcal{M} \colon \mathcal{U}^\mathbb{N} \rightarrow \mathcal{R}$ denote any mechanism satisfying $\mu$-GDP. 
    Then for any (randomized) function $g \colon \mathcal{R} \rightarrow \mathcal{R}'$ the composed mechanism $g \circ \mathcal{M} \colon \mathcal{U}^\mathbb{N} \rightarrow \mathcal{R}'$ also satisfies $\mu$-GDP.
\end{lemma}

\section{Algorithm}
\label{sec:algo}

In this section we present our main contribution. 
The pseudo-code for our mechanism is in Algorithm~\ref{alg:correlatedGaussian}.
Throughout the section we present multiple equivalent representations of our mechanism.
Each of these representations has its benefits.
The error guarantees are trivial to analyze in first representation while the privacy guarantees are easier to analyze in the second.
We present a third representation of the mechanism in Corollary~\ref{cor:covariance-matrix}, and finally in Section~\ref{sec:reduced-correlated-noise} we use yet another representation when discussing the trade-off between two sources of error.
It is sometimes subjective which representation is preferred for a specific task. 
We hope that by presenting our work in multiple ways, the core idea behind our technique will become more intuitive to readers with various backgrounds.

\begin{figure}[H]
    \centering
    \begin{minipage}{0.75\linewidth}    
    \begin{algorithm2e}[H]
        \label{alg:correlatedGaussian}
        \caption{Gaussian Mechanism using Correlated Noise}
        \SetKwInOut{Input}{Input}
        \SetKwInOut{Output}{Output}
        
        \Input{Dataset $X = (x_1, x_2, \dots, x_{n})$ where $x_i \in [0,1]^d$.}
        \Output{$\mu$-GDP estimates of $f(X) \coloneqq \sum_{i=1}^n x_i$ and $n$ under add/remove.}

        Sample $\eta \sim \mathcal{N}(0, \frac{\sqrt{d} + 1}{4 \mu^2})$.
        
        \ForEach{$i \in [d]$}
        {
          Sample $z_i \sim \mathcal{N}(0, \frac{d + \sqrt{d}}{4\mu^2})$. \\
          Let $\tilde{x}_i \leftarrow f(X)_i + \eta + z_i$. \label{line:add-noise-samples} 
        }
        
        \Return{$\tilde{x}, n + 2\eta$.}
    \end{algorithm2e}
    \end{minipage}    
\end{figure}

The motivating problem of our work is to release $d$ counting queries under differential privacy.
However, we present our mechanism in the more general setting of answering $d$ queries with a real value from the interval $[0,1]$.
That is, a dataset $X$ 
consists of $n$ data points such that $x_i \in [0, 1]^d$. 
A pair of datasets, $X \sim X'$, are neighboring if we can obtain $X'$ from $X$ by either adding or removing one data point.
As such, we either have that (1) $\vert X' \vert = n - 1$ and $f(X) - f(X') \in [0,1]^d$ or (2) $\vert X' \vert = n + 1$ and $f(X) - f(X') \in [-1,0]^d$.

The core idea of our mechanism is to take advantage of the structure of this sensitivity space.
The $\ell_2$ sensitivity of $f$ is $\sqrt{d}$.
However, notice that the vector $f(X) - f(X')$ always has $\ell_2$ distance at most $\sqrt{d}/2$ to either ${\bf\frac{1}{2}}^d$ or $-{\bf\frac{1}{2}}^d$ (recall that ${\bf\frac{1}{2}}^d$ denotes the $d$ dimensional vector where all coordinates are $1/2$).
Our mechanism spends part of the privacy budget 
on adding noise along ${\bf 1}^d$. 
The idea is to hide the difference between $f(X) - {\bf\frac{1}{2}}^d$, $f(X)$, and $f(X) + {\bf\frac{1}{2}}^d$.
This allows us to add noise with smaller magnitude to each query.
A visualization of the geometric intuition behind this idea for $d=2$ is shown in Figure~\ref{fig:geometric-intuition}.
Next, we discuss the properties of our mechanism.

\begin{figure}[t]
    \centering
    \includegraphics[width=0.975\linewidth]{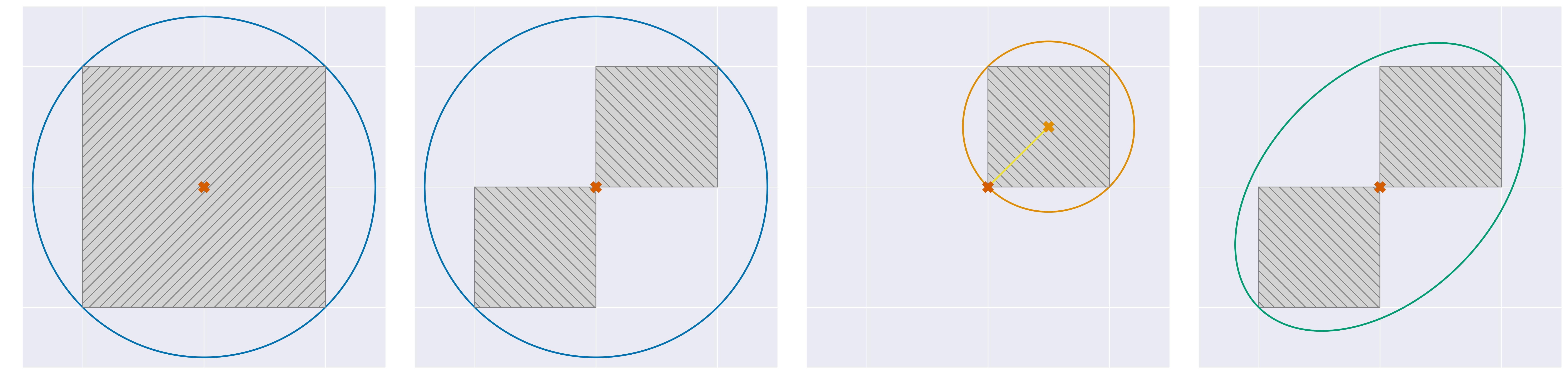} 
    \caption{2D geometric intuition of our technique. 
    Each data point is in $[0,1] \times [0,1]$. 
    The red cross marks $f(X)$ in all figures. 
    The two left figures depict the shape of noise for the standard Gaussian mechanism. 
    Informally, the mechanism is private if $f(X')$ always is inside the blue circle with radius $\sqrt{2}$.
    The value of $f(X')$ under the replacement neighboring relation can be anywhere in the grey box in the first figure.
    That is, the box shows the result of adding any point from the sensitivity space to $f(X)$.
    The boxes in the second figure similarly depict the sensitivity space under the add/remove neighboring relation.
    In general, there are $2^d$ possible values for $f(X')$ that touch the hypersphere under replacement, but only $2$ such values under add/remove for any $d$.
    In the third figure we focus on the case where we add a data point to $X$ to construct $X'$.
    Notice that the orange circle with radius $\sqrt{2}/2$ centered at $f(X) + (0.5, 0.5)$ contains the box.
    We spend part of the privacy budget to add noise along the yellow diagonal.
    The resulting noise is elliptical and 
    contains the sensitivity space under add/remove as seen in the last figure.
    } 
    \label{fig:geometric-intuition}
\end{figure}

\begin{lemma}
    \label{lem:alg1-error}
    Algorithm~\ref{alg:correlatedGaussian} estimates each query with error distributed as $\normal{\frac{d + 2\sqrt{d} + 1}{4\mu^2}}$.
\end{lemma}

\begin{proof}
    It is easy to see on Line~\ref{line:add-noise-samples} of Algorithm~\ref{alg:correlatedGaussian} that the error for the estimate of the $i$'th  query depends only on the value of $\eta$ and $z_i$ for any $i \in [d]$.
    Since $\eta$ and $z_i$ are independent samples from zero-mean Gaussians their sum is distributed as a zero-mean Gaussian with variance $\frac{d + \sqrt{d}}{4\mu^2} + \frac{\sqrt{d} + 1}{4\mu^2} = \frac{d + 2\sqrt{d} + 1}{4\mu^2}$.
    Furthermore, the standard deviation is $\left(({d + 2\sqrt{d} + 1})/({4\mu^2})\right)^{1/2} = (\sqrt{d} + 1)/(2\mu)$.
\end{proof}

It remains for us to prove the privacy guarantees of Algorithm~\ref{alg:correlatedGaussian}.
To do so, we first introduce another mechanism.
We use a simple injective function to map each data point of $X$ to a point in $\mathbb{R}^{d + 1}$.
We then release an estimate of the sum of all data points in this new representation under $\mu$-GDP using the standard Gaussian mechanism.
Finally, we prove that Algorithm~\ref{alg:correlatedGaussian} is equivalent to post-processing the output of this new mechanism for a fixed parameter.

We construct from $X$ a new dataset $\hat{X} = (\hat{x}_1, \hat{x}_2, \dots, \hat{x}_n)$ of $n$ data points where $\hat{x}_i \in \mathbb{R}^{d + 1}$.
For each $i \in [n]$ and $j \in [d]$ we define $\hat{x}_{i,j} = 2 \cdot x_{i, j} - 1$ and we let $\hat{x}_{i, d + 1} = C$ for some parameter $C > 0$.
Finally, we define $g$ as the function that computes the sum of all data points in this new representation such that $g(X) = \sum_{i = 1}^n \hat{x}_i$.
The following two lemmas show that we can estimate $g(X)$ under $\mu$-GDP using the standard Gaussian mechanism.

\begin{lemma}
    \label{lem:l2sens-in-injection}
    The $\ell_2$ sensitivity of $g$ is $\Delta g = \sqrt{d + C^2}$.
\end{lemma}

\begin{proof}
    Fix any pair of neighboring datasets $X \sim X'$ and let $x_i$ denote the data point that is either added to or removed from $X$ to obtain $X'$.
    It is easy to see that either $g(X) - g(X') = \hat{x}_i$ or $g(X') - g(X) = \hat{x}_i$.
    By the definition of $\hat{x}_i$ we have that $\hat{x}_{i,j}$ is in the interval $[-1, 1]$ for any $j \in [d]$.
    The $\ell_2$ sensitivity of $g$ is thus bounded by 
    $\|g(X) - g(X')\|_2 = \|\hat{x}_i\|_2 = \sqrt{\sum_{j = 1}^{d + 1} \hat{x}_{i,j}^2} = \sqrt{\sum_{j = 1}^{d} (2 \cdot x_{i,j} - 1)^2 + C^2} \leq \sqrt{\sum_{j = 1}^d 1^2 + C^2} = \sqrt{d + C^2}$.
\end{proof}

\begin{lemma}
    \label{lem:release-of-injection}
    The mechanism that outputs $g(X) + Z$ where $Z \sim \normal{\frac{d + C^2}{\mu^2}I_{d + 1}}$ satisfies $\mu$-GDP.
\end{lemma}

\begin{proof}
    It follows directly from Lemmas~\ref{lem:gaussian-mech-GDP} and~\ref{lem:l2sens-in-injection}.
\end{proof}

Next, we show that Algorithm~\ref{alg:correlatedGaussian} and the mechanism from Lemma~\ref{lem:release-of-injection} are equivalent under post-processing when setting $C = d^{1/4}$.
We set $C = d^{1/4}$ because that choice of $C$ minimizes the error for each query.
We discuss the effect of using other values for the parameter in more detail in Section~\ref{sec:reduced-correlated-noise}.

We recover a private estimate of $f(X)$ from the private estimate of $g(X)$ by inverting the transformation we performed when constructing $\hat{X}$ from $X$.
Notice that $f(X)_i = (g(X)_i - g(X)_{d+1}/C)/2$ which allows us to estimate $f(X)_i$ using the private estimates of $g(X)_i$ and $g(X)_{d+1}$.
Since post-processing does not affect the privacy guarantees (Lemma~\ref{lem:post-processing}), 
it follows directly from Lemma~\ref{lem:release-of-injection} and \ref{lem:equivalent-under-post-processing} that Algorithm~\ref{alg:correlatedGaussian} satisfies $\mu$-GDP.

\begin{lemma}
    \label{lem:equivalent-under-post-processing}
    Let $\mathcal{M}$ denote Algorithm~\ref{alg:correlatedGaussian} and let $\mathcal{M}'$ denote the mechanism described in Lemma~\ref{lem:release-of-injection} with parameter $C = d^{1/4}$.
    Then there exists a bijection $h \colon \mathbb{R}^d \rightarrow \mathbb{R}^{d + 1}$ such that $h(\mathcal{M}(X))$ is distributed as $\mathcal{M}'(X)$ and $h^{-1}(\mathcal{M}'(X))$ is distributed as $\mathcal{M}(X)$.
\end{lemma}

\begin{proof}
    By definition of the function $g$ we have that $g(X)_i = 2 \cdot f(X)_i - n$ for all $i \in [d]$ and $g(X)_{d + 1} = C n = d^{1/4}n$.
    Let $\tilde{x}, \tilde{n}$ denote the output of $\mathcal{M}(X)$.
    Then for $i \in [d]$ we estimate the $i$'th value of $g(X)$ as 
    \[
    h(\mathcal{M}(X))_i = 2 \cdot \tilde{x}_i - \tilde{n} = 2 \cdot (f(X)_i + \eta + z_i) - (n + 2\eta) = g(X) + 2 z_i \,.
    \]
    Since $z_i \sim \normal{\frac{d + \sqrt{d}}{4\mu^2}}$ we have that $2 \cdot z_i$ and $Z_i$ are both distributed as $\normal{\frac{d + C^2}{\mu^2}} = \normal{\frac{d + \sqrt{d}}{\mu^2}}$.
    We can estimate the final entry of $g(X)$ as 
    \[
    h(\mathcal{M}(X))_{d + 1} = d^{1/4}\tilde{n}=d^{1/4}(n + 2\eta) = g(X)_{d + 1} + d^{1/4}2\eta \,.
    \]
    Both $d^{1/4}2\eta$ and $Z_{d + 1}$ are distributed as $\normal{\frac{d + \sqrt{d}}{\mu^2}}$.
    Since the noise terms of all these estimates are independent $h(\mathcal{M}(X))$ is distributed as $\mathcal{M}'(X)$.
    
    Similarly, if $\tilde{y}$ denotes the output of $\mathcal{M}'(X)$
    we can estimate $f(X)$ and $n$ using the inverse of $h$.
    We estimate the value of $n$ as 
    \[
    h^{-1}(\mathcal{M}'(X))_{d + 1}= \tilde{y}_{d + 1}/d^{1/4} = (g(X)_{d+1} + Z_{d+1})/d^{1/4} = n + Z_{d + 1}/d^{1/4} \,,
    \] 
    and we estimate the $i$'th entry of $f(X)$ as 
    \[
    h^{-1}(\mathcal{M}'(X))_{i} = \tilde{y}_i / 2 + \tilde{y}_{d + 1}/(2d^{1/4}) = (g(X)_{i} + Z_i) / 2 + (g(X)_{d+1} + Z_{d+1})/(2d^{1/4}) = f(X)_i + Z_i / 2 + Z_{d + 1}/(2d^{1/4}) \,,
    \]
    where $Z_i / 2$ is distributed as $\normal{\frac{d + \sqrt{d}}{4\mu^2}}$.
    As such, $h^{-1}(\mathcal{M}'(X))$ has the same distribution as $\mathcal{M}(X)$ where $Z_{d + 1}/(2d^{1/4}) = \eta$ is the correlated noise added to all entries in Algorithm~\ref{alg:correlatedGaussian}.
\end{proof}

We are now ready to summarize the properties of our mechanism.

\begin{theorem}
    \label{thm:main-theorem}
    Let $f(X) \coloneqq \sum_{i = 1}^n x_i$ denote the sum of the rows in a dataset $X \in [0,1]^{n \times d}$. 
    The mechanism that outputs $f(X) + \eta {\bf 1}^d + Z$ and $n + 2\eta$ where $\eta \sim \normal{\frac{\sqrt{d} + 1}{4\mu^2}}$ and $Z \sim \normal{\frac{d + \sqrt{d}}{4\mu^2}I_d}$ satisfies $\mu$-GDP  under the add/remove neighboring relation.
    The error of each query is distributed as a zero-mean Gaussian with standard deviation $(\sqrt{d} + 1)/(2\mu)$.
\end{theorem}

\begin{proof}
    The mechanism is Algorithm~\ref{alg:correlatedGaussian}. 
    The error guarantees follow from Lemma~\ref{lem:alg1-error}.
    The privacy guarantees follow from Lemmas~\ref{lem:release-of-injection}, and \ref{lem:equivalent-under-post-processing}.
\end{proof}

Algorithm~\ref{alg:correlatedGaussian} relies on $d + 1$ independent samples from Gaussian distributions.
Alternatively, we can view the mechanism as outputting an estimate of $f(X)$ and $n$ by sampling a $d + 1$ dimensional multivariate Gaussian distribution.
The entries of the covariance matrix contain one of $4$ distinct values depending on (1) whether or not it is on the diagonal and (2) whether or not it is in the final row or column.
This representation might be useful for data analysts for estimating the effect of noise.
We use this representation again in Section~\ref{sec:reduced-correlated-noise} where we discuss different values for the parameter $C$.

\begin{corollary}
    \label{cor:covariance-matrix}
    Let $X$, $n$, and $f(X)$ be defined as in Theorem~\ref{thm:main-theorem} except that $f(X)_{d + 1} = n$.
    Then the mechanism that outputs $f(X) + Z$ where $Z \sim \normal{\Sigma/\mu^2}$ satisfies $\mu$-GDP under the add/remove neighboring relation where the covariance matrix $\Sigma \in \mathbb{R}^{(d + 1) \times (d + 1)}$ has the following structure
    \begingroup
    \renewcommand*{\arraystretch}{1.5}
    \[
        \Sigma = 
        \begin{bmatrix}
            \frac{d + 2\sqrt{d} + 1}{4} & \frac{\sqrt{d} + 1}{4} & \dots  & \frac{\sqrt{d} + 1}{4} & \frac{\sqrt{d} + 1}{2} \\
            \frac{\sqrt{d} + 1}{4} & \frac{d + 2\sqrt{d} + 1}{4} & \dots  & \frac{\sqrt{d} + 1}{4} & \frac{\sqrt{d} + 1}{2} \\
            \vdots & \vdots & \ddots  & \vdots & \vdots \\
            \frac{\sqrt{d} + 1}{4} & \frac{\sqrt{d} + 1}{4} & \dots & \frac{d + 2\sqrt{d} + 1}{4} & \frac{\sqrt{d} + 1}{2} \\
            \frac{\sqrt{d} + 1}{2} & \frac{\sqrt{d} + 1}{2} & \dots & \frac{\sqrt{d} + 1}{2} & \sqrt{d} + 1
        \end{bmatrix}    
    \]
    \endgroup
\end{corollary}

\paragraph*{Note for practitioners}
In this section we have presented multiple representations of our mechanism.
In Lemma~\ref{lem:equivalent-under-post-processing} we showed that Algorithm~\ref{alg:correlatedGaussian} is equivalent to applying the standard Gaussian mechanism to a different representation of the dataset which we denote $\hat{X}$.
However, both of these mechanisms are defined for arbitrary real-valued input and output, and as such any implementation of either representation must be an approximation of the idealized algorithm.

Real values are often approximated in computers using floating-point representation and it is well-known that naive implementations can break the formal privacy guarantees~\cite{mironov-floating-point-issues}.
Additive noise mechanisms such as the Laplace mechanism~\cite{dworkCalibratingNoise} and the Gaussian mechanism~\cite{DworkKMMN06OurDataOurselves} are typically defined as releasing the sum of a non-private query output and a sample from a distribution.
This can be problematic if the mechanism is implemented using floating-point arithmetics because the output domains might differ between neighboring datasets~(see \cite[Theorem~2]{DP-floating-point-attacks}).
To mitigate this issue \cite{DP-floating-point-attacks} present a technique for mechanisms that add a single noisy value to the true output.

However, in Line~\ref{line:add-noise-samples} of Algorithm~\ref{alg:correlatedGaussian} we add two noise samples to the query output which likely causes similar issues if floating-point representations are used.
To the best of our knowledge, it has not been studied whether implementations of mechanisms that rely on elliptical Gaussian noise suffer from floating-point vulnerabilities.
As such, implementing the version of our mechanism
from Lemma~\ref{lem:release-of-injection} might be preferred for production code.
The values of $f(X)$ and $n$ can then be estimated using post-processing as described in the proof of Lemma~\ref{lem:equivalent-under-post-processing}.
For counting queries and other discrete input we can also use the discrete Gaussian mechanism~\cite{discrete-Gaussian} to estimate $g(X)$, 
although we might need to change the value of $C$ slightly when $d^{1/4}$ is not an integer or rational number.
In Section~\ref{sec:reduced-correlated-noise} we discuss the effect of the value of $C$ on the private estimates.

\subsection{More accurate estimate of the dataset size}
\label{sec:reduced-correlated-noise}

Our mechanism outputs estimates of both $f(X)$ and the size of the dataset $n = \vert X \vert$.
So far we have ignored the accuracy of the estimate for $n$.
However, some tasks require an accurate estimate or $n$.
Here we discuss how we can release a more accurate estimate of $n$ and reduce the correlated noise at the cost of increasing the magnitude of the independent noise added to each query.

An alternative way of viewing our mechanism  
is that instead of estimating the value of each query $f(X)_i$ independently, we privately estimate the difference between $f(X)_i$ and $n/2$ for each $i \in [d]$.
We can estimate this value more accurately 
because $f(X)_i$ has sensitivity $1$ but $f(X)_i - n/2$ only has sensitivity $1/2$. 
Given an estimate of $n/2$, we can then recover an estimate of $f(X)_i$.
Let $\tilde{x}$, $\tilde{n}$ denote the output of Algorithm~\ref{alg:correlatedGaussian}.
If we subtract $\tilde{n}/2$ from all outputs we find that $\tilde{x}_i - \tilde{n}/2 = f(X)_i - n/2 + z_i$ is an estimate of the difference between $f(X)_i$ and $n/2$ and $\tilde{n} - \tilde{n}/2 = n/2 + \eta$ is our estimate of $n/2$.
The value of $n/2$ itself has a sensitivity of $1/2$. 
If we computed a new estimate for $n/2$ for each query the error would match the standard Gaussian mechanism.
The advantage of our approach is that we can reuse the same private estimate of $n/2$ for all $d$ queries.

It is easy to see that
the error for the $i$'th query using this representation depends on our estimates of both $f(X)_i - n/2$ and $n/2$. 
We can balance the amount of noise for each estimate
which is reflected in the parameter $C$ in the representation of $\hat{X}$ introduced earlier in this section.
The magnitude of noise we use to estimate the value of $f(X)_i - n/2$ is exactly $C$ times larger than the noise for estimating $n/2$. 
The following lemma gives the error for the final estimate of $f(X)_i$.

\begin{lemma}
    \label{lem:error-function-of-c}
    Let $\mathcal{M}(X) \coloneqq g(X) + Z$ where $g(X) = \sum_{i = 1}^n \hat{x}_i$ and $Z \sim \normal{\frac{d + C^2}{\mu^2}I_{d+1}}$.
    Then $\mathcal{M}$ satisfies $\mu$-GDP under the add/remove neighboring relation.
    Furthermore, if we estimate the entries of $f(X)$ by post-processing the output of $\mathcal{M}(X)$ similar to Lemma~\ref{lem:equivalent-under-post-processing} the total variance of the noise for each query is $(d + C^2 + d/C^2 + 1)/(4\mu^2)$. 
\end{lemma}

\begin{proof}
    The privacy guarantee follows from Lemma~\ref{lem:release-of-injection}.
    Let $B = (d + C^2)/\mu^2$ and $A = B/C^2 = (d/C^2 + 1)/\mu^2$.
    Then the variance of our estimate of $n/2$ (that is $\mathcal{M}(X)_{d+1}/(2C)$) is $A/4$.
    The variance of $f(X)_i - n/2$ (that is $\mathcal{M}(X)_{i}/2$) is $B/4$.
    If we estimate $f(X)_i$ as $\mathcal{M}(X)_i/2 + \mathcal{M}(X)_{d+1}/(2C)$ then the variance is $(A + B)/4$. 
\end{proof}

Notice that the dependence on $C$ for the variance in Lemma~\ref{lem:error-function-of-c} is $(C^2 + d/C^2)/4$.
Since $C > 0$, the global minimum of $(C^2 + d/C^2)/4$ is at $C=d^{1/4}$ which is the value we use in Lemma~\ref{lem:equivalent-under-post-processing} and for Algorithm~\ref{alg:correlatedGaussian}.
The variable $A$ used in the proof is decreasing in $C$ which implies that we can get a more accurate estimate of the dataset size by increasing the value of $C$.
However, this will in turn increase the error of the estimate of $f(X)_i - n/2$ because the variable $B$ is increasing in $C$.

Figure~\ref{fig:covariance-d-10000} shows the effect of varying $C$ when $d=10000$.
As stated above the total variance of each query is minimized for $C=d^{1/4}=10$.
However, the accuracy of our estimate of $n$ improves as $C$ increases.
The plot shows that we can get a much more accurate estimate of $n$ while still outperforming the standard Gaussian mechanism in terms of the total variance.
Even if we set $C=d^{1/2}=100$ such that we get an estimate of $n$ with variance only $A = 2/\mu^2$ the variance for each query is $(d + 1)/(2\mu^2)= 5000.5/\mu^2$ which is roughly half when compared with $d/\mu^2 = 10000/\mu^2$ for the standard Gaussian mechanism.

For many problems the choice of $C = d^{1/4}$ is preferred because it gives the best estimate of $f(X)_i$.
We get no benefit from setting the parameter $C < d^{1/4}$ because the increased magnitude of the correlated noise results in worse estimates both of $n$ and $f(X)_i$.
However, if we want a better estimate of the dataset size and reduce the magnitude of correlated noise we can increase the value of $C$. 
This can be beneficial for tasks that heavily rely an the accuracy of the estimate for $n$. 
It can also be useful if the downstream impact of correlated noise is either difficult to analyze or control.
Due to diminishing returns we likely never want to use a value of $C$ larger than $\sqrt{d}$.
By setting $C = \sqrt{d}$ we already get an estimate of the dataset size that is independent of $d$ as discussed above.
The optimal choice of $C$ depends on the exact use case.

\begin{figure}[t]
    \centering
    \begin{minipage}{0.59\linewidth} 
    \includegraphics[width=\linewidth]{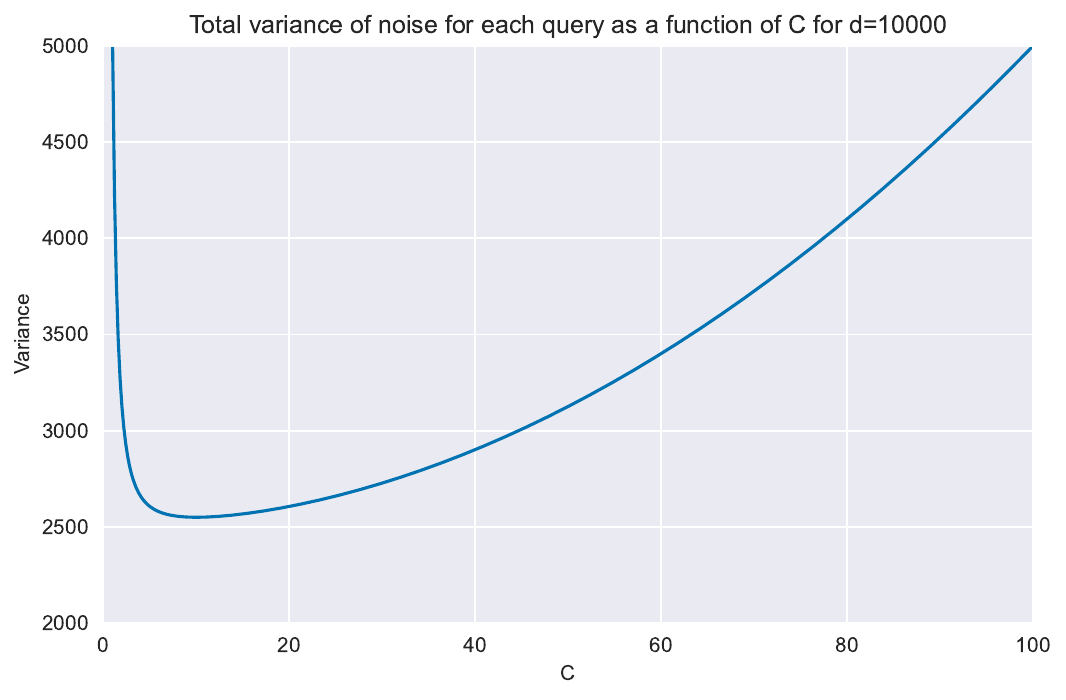} 
    \end{minipage}
    \hfill
    \begin{minipage}{0.40\linewidth}
        \begingroup
        \renewcommand*{\arraystretch}{1.5} 
        \[
        \Sigma = 
        \begin{bmatrix}
            \frac{A + B}{4} & \frac{A}{4} & \dots  & \frac{A}{4} & \frac{A}{2} \\
            \frac{A}{4} & \frac{A + B}{4} & \dots  & \frac{A}{4} & \frac{A}{2} \\
            \vdots & \vdots & \ddots  & \vdots & \vdots \\
            \frac{A}{4} & \frac{A}{4} & \dots & \frac{A + B}{4} & \frac{A}{2} \\
            \frac{A}{2} & \frac{A}{2} & \dots & \frac{A}{2} & A
        \end{bmatrix}
        \]    
        \endgroup
    \end{minipage}
    \caption{
    The plot shows the effect of changing the parameter $C$ for datasets with $d = 10000$.   
    On the right, we show the covariance matrix similar to Corollary~\ref{cor:covariance-matrix}.
    All entries in the matrix should be scaled by $1/\mu^2$.
    For simplicity we consider $\mu = 1$.
    We have that $A = d/C^2 + 1 = 10^4/C^2 + 1$ and $B = d + C^2 = 10^4 + C^2$.
    The plot depicts the value of $(A + B)/4$ as a function of $C$.} 
    \label{fig:covariance-d-10000}
\end{figure}

In this subsection we discussed the trade-off between the private estimate for the size of the input dataset and the error of each query.
Throughout this paper we have assumed that we are not given an estimate of $n$ and we explicitly use part of our privacy budget to find such an estimate.
However, sometimes we actually have access to an estimate of $n$.
This is for example the case if the mechanism is used as a subroutine in a bigger system where we already computed a private estimated the dataset size.
In such cases, we are not required to spend any privacy budget on finding a new estimate.

\begin{lemma}
    \label{lem:already-have-estimate}
    Let $f(X) \coloneqq \sum_{i = 1}^n x_i$ denote the sum of the rows in a dataset $X \in [0,1]^{n \times d}$.
    Assume that we have access to an estimate of the dataset size which we denote $\tilde{n}$.
    Then the mechanism that outputs $f(X) - n{\bf 1}^d/2 + \tilde{n}{\bf 1}^d/2 + Z$  where $Z \sim \normal{\frac{d}{4\mu^2}I_d}$ satisfies $\mu$-GDP under the add/remove neighboring relation.
\end{lemma}

\begin{proof}
    It follows from Lemma~\ref{lem:gaussian-mech-GDP} because the $\ell_2$ sensitivity of $g(X) \coloneqq f(X) - n{\bf 1}^d/2$ is $\sqrt{d}/2$.
\end{proof}

\section{Other settings}
\label{sec:extensions}

In the previous section we presented a simple technique for added elliptical noise for the fundamental task of privately answering $d$ queries from $[0,1]$ under the add/remove neighboring relation.
This problem setting was the primary motivation for our work.
However, we believe that our technique has the potential to improve the accuracy of differentially private mechanisms for other settings as well.
To justify this claim we consider an extension of the setting from the previous section.
We improve the variance compared to the standard Gaussian mechanism for this problem under both the add/remove and the replacement neighboring relation.
We end the section by comparing our technique with a recent result from related work.
In Section~\ref{sec:conclusion} we discuss potential future extensions.

We consider a setting where a dataset $X$ consists of $n$ data points where $x_i \in [0,1]^{m \times d}$. 
The input is restricted such that each data point has non-zero entries in only a single row, that is $\vert \{j : x_{i, j} \neq {\textbf 0}^d \}\vert \leq 1$.
As such we have that $\|x_i\|_2 = (\sum_{j = 1}^m \sum_{k = 1}^d x_{i,j,k}^2)^{1/2} \leq \sqrt{d}$.
This kind of simple hierarchical data 
is relevant if we are answering $d$ queries and the outputs are grouped based on some property of the data points.
Two such examples would be to group users based on location or time when their data was reported.
We first consider how the standard Gaussian mechanism estimates $f(X) \coloneqq \sum_{i = 1}^n x_i$ under either neighboring relation.

\begin{lemma}
    \label{lem:standard-gaussian-with-buckets}
    The mechanism that adds noise independently to each entry of $f(X)$ distributed as $\normal{d/\mu^2}$ satisfies $\mu$-GDP under the add/remove neighboring relation.
    The same mechanism satisfies $\mu$-GDP under the replacement neighboring relation when the noise is distributed as $\normal{2d/\mu^2}$.
\end{lemma}

\begin{proof}
    Since $\|x_i\|_2 \leq \sqrt{d}$ the $\ell_2$ sensitivity of $f(X)$ is clearly $\sqrt{d}$ under the add/remove neighboring relation.
    Under the replacement neighboring relation $f(X)$ and $f(X')$ differs in at most $2d$ entries by at most $1$ and as such $\|f(X) - f(X')\|_2 \leq \sqrt{2d}$. 
    The privacy guarantees then follow from Lemma~\ref{lem:gaussian-mech-GDP}.
\end{proof}

\begin{figure}
    \centering
    \begin{minipage}{0.8\linewidth}    
    \begin{algorithm2e}[H]
        \label{alg:correlatedGaussianGroups}
        \caption{Correlated Gaussian Mechanism for disjoint queries}
        \SetKwInOut{Input}{Input}
        \SetKwInOut{Output}{Output}
        
        \Input{Dataset $X = (x_1, x_2, \dots, x_{n})$ where $x_i \in [0,1]^{m \times d}$ and for each $i \in [n]$ it holds that $\vert \{j : x_{i, j} \neq {\textbf 0}^d \}\vert \leq 1$.}
        \Output{$\mu$-GDP estimates of $f(X) \coloneqq \sum_{i=1}^n x_i$ and the count for each row.}

        \ForEach{$j \in [m]$}
        {
          Sample $\eta_j \sim \begin{cases}
            \mathcal{N}(0, \frac{\sqrt{d} + 1}{4\mu^2}) & \text{under the add/remove neighboring relation,} \\
            \mathcal{N}(0, \frac{4}{\mu^2}) & \text{under the replacement neighboring relation.}
        \end{cases}$\\
          \ForEach{$k \in [d]$}
          {
          Sample $z_{j, k} \sim \begin{cases}
            \mathcal{N}(0, \frac{d + \sqrt{d}}{4\mu^2}) & \text{under the add/remove neighboring relation,} \\
            \mathcal{N}(0, \frac{d}{\mu^2}) & \text{under the replacement neighboring relation.}
            \end{cases}$ \\
          Let $\tilde{x}_{j,k} \leftarrow f(X)_{j, k} + \eta_j + z_{j, k}$.
          } 
          $\tilde{y}_j \leftarrow \vert \{x_i : x_{i, j} \neq {\textbf 0}^d \}\vert + 2\eta_j$.\\
        }
            
        \Return{$\tilde{x}$, $\tilde{y}$.}
    \end{algorithm2e}
    \end{minipage}    
\end{figure}

Next, we consider how to adapt our mechanism to this setting.
The main idea of Algorithm~\ref{alg:correlatedGaussian} is to add the same noise sample to all queries.
This allows us to reduce the magnitude of the independent noise samples for each query when the sensitivity is concentrated along ${\bf 1}^d$ as shown for $d = 2$ in Figure~\ref{fig:geometric-intuition}.
In this setting we do not benefit much from adding the same noise sample to all queries.
Instead we sample correlated noise for each row.
We also release an estimate of the number of data points with non-zero entries for each row instead of releasing an estimate of $n$.
The pseudo-code is shown in Algorithm~\ref{alg:correlatedGaussianGroups}. We prove the privacy guarantees next.

\begin{lemma}
    \label{lem:correlated-noise-groups}
    Algorithm~\ref{alg:correlatedGaussianGroups} satisfies $\mu$-GDP.
\end{lemma}

\begin{proof}
    In the case of the add/remove relation $f(X)$ and $f(X')$ only differ in one row.
    Algorithm~\ref{alg:correlatedGaussianGroups} essentially runs Algorithm~\ref{alg:correlatedGaussian} independently for each row. The privacy guarantees follow from the same proof as Lemma~\ref{lem:equivalent-under-post-processing}.

    We prove that the privacy guarantees holds under replacement using an alternative representation of $X$ similar to Lemma~\ref{lem:release-of-injection} for the setting of Section~\ref{sec:algo}.
    Define $\hat{x}_i \in \mathbb{R}^{m \times (d + 1)}$ such that $\hat{x}_{i,j,k} = 0$ for all $k \in [d+1]$ if $\|x_{i, j}\|_2 = 0$.
    By definition of the problem setup we have $\|x_{i, j}\|_2 = 0$ for all but one row for each $i \in [n]$.
    For the row $j \in [m]$ where $\|x_{i, j}\|_2 > 0$ we define
    $\hat{x}_{i,j,k} = 2 \cdot x_{i,j,k} - 1$, and $\hat{x}_{i,j,d + 1} = d^{1/2}$. We define the function $g(X) \coloneqq \sum_{i = 1}^n \hat{x}_i$.
    
    Let $x_i$ denote the data point we replace with $x'_i$ to obtain $X'$ from $X$ and fix the row $j \in [m]$ such that $\| x_{i, j} \|_2 > 0$.
    We have two cases to consider 
    (1) If the non-zero entries of $x'_i$ are also in row $j$ (that is, $\| x'_{i,j} \|_2 > 0$) then $\hat{x}_{i,j,k}$ and $\hat{x}'_{i,j,k}$ differ by at most $2$ for each $k \in [d]$ for row $j$ and all other entries of $\hat{x}_i$ are unchanged.
    As such we have $\|g(X) - g(X')\|_2 = \|\hat{x}_{i,j} - \hat{x}'_{i,j}\|_2 \leq \sqrt{d \cdot 2^2} = 2\sqrt{d}$.
    (2) If the non-zero entries of $x'_i$ are in another row $j'$ the entries of $g(X)$ and $g(X')$ differ only in these two rows.
    Each entry is zero for at least one of $\hat{x}_i$ and $\hat{x}'_i$.
    In this case we have that $\|g(X) - g(X')\|_2 = \|\hat{x}_{i} - \hat{x}'_{i}\|_2 \leq (2 \cdot (\sum_{k = 1}^d 1^2 + (d^{1/2})^2))^{1/2} = \sqrt{4d} = 2\sqrt{d}$.

    As such, the $\ell_2$ sensitivity of $g$ under the replacement neighboring relation is $2\sqrt{d}$ and we can release $g(X)$ under $\mu$-GDP by adding noise from $\normal{4d/\mu^2}$ independently to each entry.
    The proof that Algorithm~\ref{alg:correlatedGaussianGroups} is equivalent to post-processing the above mechanism follows the steps of the proof for Lemma~\ref{lem:equivalent-under-post-processing} applied to each row for $C = d^{1/2}$.
\end{proof}

For certain tasks we want an accurate estimate of the number of data points in each row.
We can get a better estimate of the count for data points in each row ($\tilde{y} \in \mathbb{R}^m$) under the add/remove neighboring relation by balancing the different sources of noise similar to Section~\ref{sec:reduced-correlated-noise}.

\subsection{Bounded density}
\label{sec:bounded-density}

\begin{figure}[t]
    \centering
    \hfill
    \begin{minipage}{0.45\linewidth} 
    \includegraphics[width=\linewidth]{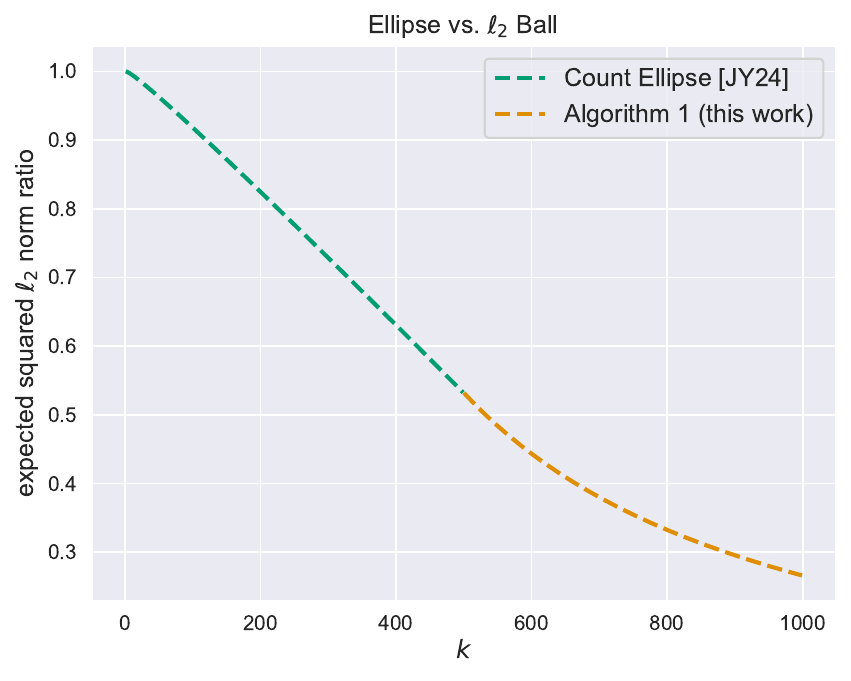} 
    \end{minipage}
    \begin{minipage}{0.45\linewidth}
        \includegraphics[width=\linewidth]{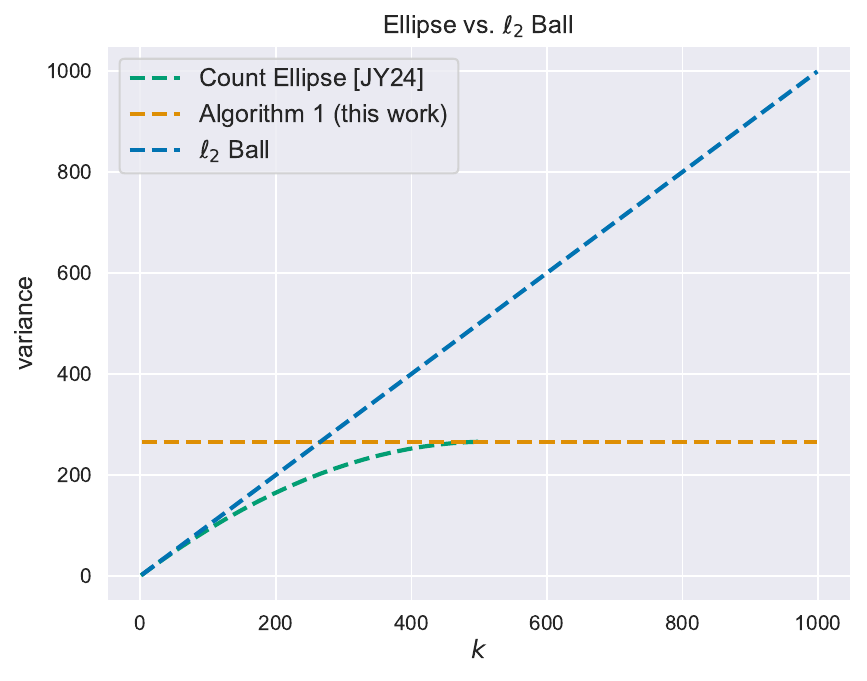}
    \end{minipage}
    \hfill
    \caption{Comparison of error for the bounded Count problem for $d=1000$. Note that the optimal Count ellipse by~\cite{constructions-k-norm-elliptic} requires $k \leq d/2$. The left plot corresponds to the right plot in~\cite[Figure~1]{constructions-k-norm-elliptic} where we included our mechanism for $k > d/2$. 
    The line for the $\ell_2$ Ball represents the standard Gaussian mechanism in this setting.} 
    \label{fig:count-sparse}
\end{figure}

Here we consider a setting where each data point contributes to a bounded number of queries.
This setting was recently studied by~\cite{constructions-k-norm-elliptic}. 
They gave efficient constructions for adding elliptical Gaussian noise to two classes of queries.
The relevant query for this paper is the \textit{Count} problem, which is equivalent to the setting from Section~\ref{sec:algo} with the additional constraint that each data point has at most $k \leq d$ non-zero entries.
Under this restricted setting the standard Gaussian mechanism adds noise to each query where the variance scales only by $k$ rather than $d$.
Similar to Algorithm~\ref{alg:correlatedGaussian}, the mechanism by \cite{constructions-k-norm-elliptic} adds additional noise along ${\bf 1}^d$.
They give closed-form solutions for the shape of the optimal ellipse minimizing expected squared $\ell_2$ norm. 
However, their result only applies when $k \leq d/2$.
The setting we considered in Section~\ref{sec:algo} corresponds to the case where $k = d$.
A natural question is whether the error can be improved for the range $d/2 < k < d$.
It turns out that there is almost no room for improvement. 
In fact, our mechanism is almost optimal for any $k \geq d/2$ as we show below.

In Figure~\ref{fig:count-sparse} we recreated a plot from Figure 1 of~\cite{constructions-k-norm-elliptic}.
They plot the error relative to the standard Gaussian mechanism. 
The improvement increases with $k$, but their mechanism is undefined for $k > d/2$.
We extended the line with our mechanism for $k > d/2$ showing the improvement from elliptical noise for all values of $k$.
We also plot all three mechanisms without normalizing the variance.
From that plot it is clear that their mechanism should be used for small $k$ while our mechanism is preferred for any $k > d/2$.
Next, we bound the difference in variance between our mechanism and the optimal general Gaussian mechanism for any $d/2 \leq k \leq d$.

We first restate a result of~\cite{constructions-k-norm-elliptic} since the variance of their mechanism is optimal among Gaussian mechanisms for any $k \leq d/2$. 
Here we are interesting in the variance of their mechanism with the least restrictive parameter, that is $k = \lfloor d/2 \rfloor$.
They do not explicitly state the variance in their paper, but we can calculate it from their closed form expressions describing the shape of the Count ellipse.
For simplicity of presentation we only consider even values of $d$. 
The equation for the variance of their mechanism when $d$ is odd (and $k = (d - 1)/2$) is more complicated, and we do not need it to show that our mechanism is close to optimal for any $d$.

\begin{lemma}
    \label{lem:JY24-variance}
    The Count ellipse mechanism by \cite{constructions-k-norm-elliptic} with parameters $k = d/2$ (where $d$ is even) and $\rho = \mu^2/2$ satisfies $\mu$-GDP.
    The mechanism estimates each query with error distributed as
    \[
        \mathcal{N}\left(0, \frac{d + 2\sqrt{d - 1}}{4\mu^2} \right) \,.
    \]
\end{lemma}

\begin{proof}
    We do not prove the privacy guarantees of their mechanism and we only sketch the proof for the variance.
    For more rigorous details about the mechanism we refer to their paper and references within.
    For simplicity, we assume without loss of generality that $\mu = 1$.
    In the proof of \cite[Theorem~9.9]{constructions-k-norm-elliptic} we find their closed form expressions for variables $a_1$ and $a_2$.
    We do not repeat the expressions here and refer interested readers to their paper.
    The mechanism adds zero-mean Gaussian noise in the shape of an ellipse with an axis along $\mathbf{1}^d$ of length $a_1$.
    The remaining $d - 1$ axis lengths are all equal to $a_2$.
    The expected squared $\ell_2$ norm of the noise is $a_1^2 + (d-1) \cdot a_2^2$,
    and the variance of the error for each query is $(a_1^2 + (d-1) \cdot a_2^2)/d$.
    By setting $k = d/2$ in the expressions for $a_1$ and $a_2$ and simplifying we have that $(a_1^2 + (d-1) a_2^2)/d = (d + 2\sqrt{d-1})/4$.
    Here we left out the steps to simplify the expression for brevity.
\end{proof}

Recall from Lemma~\ref{lem:alg1-error} that the variance of Algorithm~\ref{alg:correlatedGaussian} is $(d + 2 \sqrt{d} + 1)/(4\mu)$. 
The variance achieved by \cite{constructions-k-norm-elliptic} is optimal among general Gaussian mechanisms for any $k \leq \lfloor d/2 \rfloor$.
The variance of an optimal mechanism is non-decreasing in $k$. 
As such, the variance of our mechanism is almost optimal for any $k \geq d/2$.

\begin{lemma}
    \label{lem:near-optimal-bound}
    Let $OPT$ denote the minimum variance for any general Gaussian mechanism under $\mu$-GDP for any $d$ and parameter $k \geq d/2$.
    Then the variance of Algorithm~\ref{alg:correlatedGaussian} is at most $OPT + 2\sqrt{3}/(4\mu^2) < OPT + 0.87/\mu^2$.  
\end{lemma}

\begin{proof}
    We first consider even values of $d$ and the gap to the optimal variance for $k = d/2$.
    We then bound the error when $d$ is odd.
    If $d$ is even the difference in variance between our mechanism and the optimal Count ellipse is
    \[
    \frac{d + 2\sqrt{d} + 1}{4\mu^2} - \frac{d + 2\sqrt{d - 1}}{4\mu^2} = \frac{1 + 2\sqrt{d} - 2\sqrt{d-1}}{4\mu^2} \,.
    \]
    This function is decreasing in $d$ and as such it is maximized for $d=2$ where $(1 + 2\sqrt{2} - 2\sqrt{2-1})/(4\mu^2) \leq 0.46/\mu^2$.
    Since the variance for the Count ellipse mechanism with parameter $k = d/2$ is a lower bound on $OPT$ for any $k \geq d/2$, this is an upper bound on the gap to the optimal variance for all even $d$.

    When $d > 2$ is odd (in the special case of $d = 1$ the variance is optimal) we can bound the error using the triangle inequality. 
    The difference in variance of Algorithm~\ref{alg:correlatedGaussian} for $d$ and $d-1$ is 
    \[
        \frac{d + 2\sqrt{d} + 1}{4\mu^2} - \frac{(d-1) + 2\sqrt{(d-1)} + 1}{4\mu^2} = \frac{1 + 2\sqrt{d} - 2\sqrt{d-1}}{4\mu^2} \,.
    \]
    The value of $OPT$ when $d$ is odd and $k > d/2$ is lower bounded by the value of $OPT$ for parameters $d - 1$ and $k = (d-1)/2$ that we discussed above. 
    As such, the variance of Algorithm~\ref{alg:correlatedGaussian} is upper bounded by
    \[
        OPT + \frac{1 + 2\sqrt{d-1} - 2\sqrt{(d-1)-1}}{4\mu^2} + \frac{1 + 2\sqrt{d} - 2\sqrt{d-1}}{4\mu^2} = OPT + \frac{2 + 2\sqrt{d} - 2\sqrt{d-2}}{4\mu^2} \,.
    \]

    This gap for the bound is again decreasing in $d$ and the bound is exact for the degenerate case of $d = 3$ and $k = 1$.
    When $k = 1$ the optimal mechanism is the standard Gaussian mechanism which adds spherical Gaussian noise with variance $1/\mu^2$. We have that Algorithm~\ref{alg:correlatedGaussian} achieves variance $(4 + 2\sqrt{3})/(4\mu^2)$ for $k = d = 3$.
\end{proof}

We leave finding the exact value of $OPT$ for $k > d/2$ as an open problem.
As a small note, we point out that the mechanism of \cite{constructions-k-norm-elliptic} does not include an estimate of the dataset size. 
The optimal ellipse might differ slightly if we want to return an estimate of $n$ similar to Algorithm~\ref{alg:correlatedGaussian}.
It is also worth noting that it is likely possible to slightly reduce the variance even for $k = d$ if we do not care about the estimate of $n$.
Notice that there is a small gap between the corners of the sensitivity space and our ellipse in Figure~\ref{fig:geometric-intuition}.
If we instead add Gaussian noise in the shape of a tighter ellipse we would reduce the variance slightly.
We do not consider that in this paper because our mechanism is simpler, and we can return an estimate of the dataset size at almost no cost for larger $d$ as shown above in Lemma~\ref{lem:near-optimal-bound}.
Any improvement in variance will only be noticeable for tiny values of $d$.

\section{Related work}
\label{sec:related}

The Gaussian mechanism satisfies several definitions of differential privacy. 
Exact bounds are known of $(\varepsilon, \delta)$-DP~\cite{analyticalGaussian}, $(\alpha, \varepsilon)$-RDP~\cite{renyiDP}, $\rho$-zCDP~\cite{zeroConcentrated}, and $\mu$-GDP~\cite{gaussianDPand-f-DP}.
Our privacy guarantees follow from a reduction to the standard Gaussian mechanism.
As such, our results apply to any of the above definitions of differential privacy.
We choose to present our mechanism in terms of $\mu$-GDP because this definition exactly describes the privacy guarantees of the Gaussian mechanism.
For all problems we consider in this paper our privacy guarantees for worst-case pairs of neighboring datasets exactly match those of the standard Gaussian mechanism.
In technical terms, this means that there exist pairs of neighboring datasets for which the privacy loss random variable of our mechanism is distributed as $\mathcal{N}(\mu^2/2, \mu^2)$.
That is also the case for the standard Gaussian mechanism.
We refer to \cite[Proposition~3]{steinkeCompositionAmplification} for a discussion of the privacy loss distribution of additive Gaussian noise.

\cite{betweenSteinkeUllman} gave a lower bound on the sample complexity for answering $d$ counting queries under $(\varepsilon, \delta)$-differential privacy. 
They show that the Gaussian mechanism has asymptotical optimal $\ell_1$ error under reasonable assumptions on $\delta$. 
The input domain in their paper is $\{-1, 1\}^d$ rather than $\{0, 1\}^d$.
They only consider the replacement neighboring relation and in that setting the counting query problem is equivalent for the two domains.
It is easy to translate any estimate from one domain to the other using a simple post-processing step assuming that $0$ should be mapped $-1$.
In contrast, the problem differs between the two domains under the add/remove neighboring relation. 
We cannot directly translate an estimate from one domain to the other using post-processing since the size of datasets is not fixed.
However, the representation of the problem in the domain $[-1,1]$ plays a key role in the privacy proof for our mechanism.

The mechanism we present in this paper adds unbiased correlated Gaussian noise designed to fit the structure of the queries.
Adapting Gaussian noise to the shape of the sensitivity space is not a novel concept.
We refer to mechanisms that add (possibly correlated) unbiased noise as general Gaussian mechanisms.
\cite{geometryOfDPSparseApproximate} gave an algorithm for answering any set of linear queries under approximate differential privacy.
Linear queries are a broad class of queries that includes all problems considered in this paper.
They give an algorithm for choosing the shape of Gaussian noise for any linear query by fitting an ellipse to the convex hull of the sensitivity space.
\cite{NikolovTang24GeneralGaussian} considered the problem of mean estimation under the replacement neighborhood when data points are from some bounded domain.
They show that the general Gaussian mechanism achieves near-optimal error among all unbiased mechanisms for a broad class of error measures in that setting.
However, their algorithm for computing the shape of noise for arbitrary queries can be inefficient.

In contrast to the general results of~\cite{geometryOfDPSparseApproximate, NikolovTang24GeneralGaussian}, 
we only consider a subset of linear queries whose sensitivity space has a certain structure.
This allows us to tailor our mechanism specifically towards that structure.
Concurrent with our work, \cite{constructions-k-norm-elliptic} gave efficient constructions for adding elliptical Gaussian noise to two types of linear queries.
They present closed-form solutions for the shape of the optimal ellipse for \textit{Count} and \textit{Vote} queries.
Most relevant to this paper is the Count problem which we discuss in more detail in Section~\ref{sec:bounded-density}.
They consider the problem where entries are in $[0,b]$ rather than $[0,1]$.
Note that the task of privately estimating the sum is equivalent for both settings as we can simply scale the input by a factor of $b$.

A recent result by~\cite{add-remove-mean-estimation} improved the constant for private mean estimation under the add/remove neighboring relation.
Their technique also includes an estimate of $n$ and an embedding of the sum in a space using an additional dimension.
However, their setting differs significantly from ours and they only consider one-dimensional data making our results incomparable.

Concurrent to this work, \cite{free-lunch} also considered the task of privately estimating sums and dataset size, with the goal of outputting a private mean.
Their technique also relies on an embedding of the data points using an additional dimension.
Similar to our mechanism, they return a private estimate of the dataset size at no additional privacy cost.
They explore multiple settings in their paper.
In the case of Gaussian noise for one dimensional data ($x_i \in [0,1]$), the error of their mechanism is identically distributed to Algorithm~\ref{alg:correlatedGaussian}.
Although highly subjective, their technique is arguably preferred when $d = 1$ because their privacy proof is simple and very intuitive (see Figure 1 of their paper).
However, they do not reduce the error for the sums when $d > 1$, which is the main contribution of our work.

Our mechanism relies on an intermediate representation of the query outputs which can be privately estimated using the standard Gaussian mechanism.
An estimate of the original queries $f(X)$ can then be recovered using post-processing.
This technique is similar to the matrix mechanism~\cite{matrix-mechanism} used for optimizing the error of differentially private linear queries.
The mechanism works by factorizing a query matrix into two new matrices.
The first matrix is then applied to the input dataset and either the standard Laplace or Gaussian mechanism is used to privately release the output. 
The other matrix is then used to recover an estimate of the original query matrix.
The technique has been used to improve mechanisms in various settings such as computing prefix sums~(e.g. \cite{almost-tight-continual,online-prefix-sums,lpp-metric}).

\cite{LT23-DPMG} gave an $(\varepsilon, \delta)$-differentially private mechanism for releasing a Misra-Gries sketch.
At a high level, their mechanism is very similar to ours. 
They exploit the structure of the sensitivity space of the Misra-Gries sketch and add much less noise compared to standard techniques that scale the noise by the $\ell_1$ or $\ell_2$ sensitivity. 
They add the same noise sample to all counters in the sketch to handle a special case where all counters differ by $1$ between some pairs of neighboring datasets.
This allows them to reduce the magnitude of noise significantly because Misra-Gries sketches for neighboring datasets only differ by $1$ in a single counter except for this special case.
They present some intuition behind the privacy guarantees in the technical overview of their paper using a $d+1$ representation which is similar to the one we use in Lemma~\ref{lem:release-of-injection}.
However, the structure of the sensitivity space for the queries we consider is different than the structure for the Misra-Gries sketch.

\section{Conclusion and future work}
\label{sec:conclusion}

In this paper we presented a simple variant of the general Gaussian mechanism for the fundamental task of answering $d$ independent queries in $[0,1]$ under the add/remove neighboring relation.
In Section~\ref{sec:extensions} we extended our mechanism to a setting where queries are divided into $m$ groups that form a partition of all queries and each data point only contributes to a single group.
The general idea in this work is to take advantage of the hierarchical structure of these queries. 
Adding the same noise sample to all queries in a group allows us to reduce the noise added independently to each query.
We believe this technique can be used to design differentially private mechanisms for other problems with similar structures.
For example, we could consider settings where the groups of queries partially overlap instead of forming a partition.
Alternatively, data points might contribute to a bounded number of groups instead of only one.
In such settings we might need to increase the magnitude of correlated noise.
Another interesting avenue to explore is whether correlated noise can improve constants for the Gaussian Sparse Histogram Mechanism~\cite{GaussianSparseHistogramMechanism}.
Finally, it is worth noting that we presented our work as a variant of the Gaussian mechanism, but the technique can be adapted to other mechanisms such as the Laplace mechanism for pure differential privacy.
\\
\paragraph*{Acknowledgments}
This work was carried out when the author was at the IT University of Copenhagen.
We thank Joel Daniel Andersson and other members of the Providentia project at University of Copenhagen for helpful discussions. 
We thank anonymous reviewers for suggestions that helped improve the paper.

\bibliographystyle{alpha}
\bibliography{literature}

\end{document}